\documentclass[conference]{IEEEtran}
\ifCLASSINFOpdf
\else
\fi

\usepackage{eulervm}

\usepackage{graphicx}
\usepackage{subfigure}
\usepackage{array}
\usepackage{amsmath}
\usepackage{amssymb}
\usepackage{amsthm}
\usepackage{mathrsfs}
\usepackage{bbm}
\usepackage{bbold}
\usepackage{xfrac}
\usepackage{algorithm}
\usepackage{algorithmic}
\usepackage{upgreek}
\usepackage{clrscode3e}
\usepackage{bm}
\usepackage{cite}

\DeclareFontFamily{OT1}{pzc}{}
\DeclareFontShape{OT1}{pzc}{m}{it}{<-> s * [1.10] pzcmi7t}{}
\DeclareMathAlphabet{\mathpzc}{OT1}{pzc}{m}{it}

\begin{document}
%
\title{A Performance Study of Energy Minimization for Interleaved and Localized FDMA}


\author{\IEEEauthorblockN{Lei You$^1$, Lei Lei$^2$, and Di Yuan$^{2,3}$}
\IEEEauthorblockA{\small$^1$Information Engineering College,
Qingdao University, China\\}
\IEEEauthorblockA{\small$^2$Department of Science and Technology,
Link\"{o}ping University, Sweden\\}
\IEEEauthorblockA{\small$^3$Institute for Systems Research, University of Maryland, USA\\}
\small\texttt{youleiqdu@gmail.com, \{lei.lei, di.yuan\}@liu.se, diyuan@umd.edu}
}


%


\maketitle

\begin{abstract}
Optimal channel allocation is a key performance engineering aspect in single-carrier frequency-division multiple access (SC-FDMA). It is of significance to consider minimum sum power (Min-Power), subject to meeting specified user's demand, since mobile users typically employ battery-powered handsets. In this paper, we prove that Min-Power is polynomial-time solvable for interleaved SC-FDMA (IFDMA). Then we propose a channel allocation algorithm for IFDMA, which is guaranteed to achieve global optimum in polynomial time. We numerically compare the proposed algorithm with optimal localized SC-FDMA (LFDMA) for Min-Power. The results show that LFDMA outperforms IFDMA in the maximal supported user demand. When the user demand can be satisfied in both LFDMA and IFDMA, LFDMA performs slightly better than IFDMA. However Min-Power is polynomial-time solvable for IFDMA whereas it is not for LFDMA.
\end{abstract}


%
\IEEEpeerreviewmaketitle

\theoremstyle{plain}
\newtheorem{definition}{Definition}
\newtheorem{theorem}{Theorem}
\newtheorem{lemma}{Lemma}
\newtheorem{proposition}{Proposition}
\newtheorem{postulation}{Postulation}

\section{Introduction}
\IEEEPARstart{O}{ver} the past years orthogonal frequency division multiplexing access (OFDMA) has been an important technique for broadband wireless communications. A major advantage of OFDMA is its robustness in the presence of multi-path fading in cellular applications \cite{Nee:2000vr}. In third generation partnership project long term evolution (3GPP-LTE) standard, the uplink access scheme is single-carrier frequency-division multiple access (SC-FDMA) \cite{Myung:2006fe}, a modified version of OFDMA but having similar throughput performance and essentially the same overall complexity as OFDMA. 

There are two approaches to assign users among channels. In localized SC-FDMA (LFDMA), each user uses a set of adjacent channels to transmit data \cite{Ahmen:2011globecom}. 
The other is distributed SC-FDMA in which the channels used by a user are spread over the entire channel spectrum. 
One realization of distributed SC-FDMA is interleaved SC-FDMA (IFDMA) \cite{Sorger:1998ip} where the occupied channels for each user are equidistant from each other. Currently, IFDMA as well as LFDMA has been investigated in 3GPP-LTE for the uplink transmission \cite{Frank:2007cz}. The trade-off on channel allocation between LFDMA and IFDMA is investigated in many literatures. In \cite{Song:ga}, Song et.al. state that IFDMA has less carrier frequency offset (CFO) interference but LFDMA achieves more diversity gain. In \cite{Myung:2006fe} Myung et.al. find that LFDMA with channel-dependent scheduling (CDS) results in higher throughput than IFDMA, whereas the peak to average power ratio (PAPR) performance of IFDMA is better than that of LFDMA. 

Battery-powered equipments are increasingly employed by mobile users. It is of significance to consider minimum sum power (Min-Power), subject to meeting demand target \cite{Kivanc:2003twc,Feiten:2003twc,Joung:2012iet}. The Min-Power problem for LFDMA is proved to be $\mathcal{NP}$-hard in \cite{Lei:2013en}. As far as our information goes, few literatures investigate the Min-Power problem in IFDMA, while some heuristic algorithms for consecutive channel allocation are presented in \cite{Lei:2013en, Nam:2010ky, Kim:2010fq, Sokmen:2010gv, Wong:2009gra}. In this paper, we present a minimal power channel allocation ($\proc{MPCA}$) algorithm  for IFDMA, which achieves global optimality in polynomial time. The Min-Power in IFDMA is modeled as a combinatorial optimization problem. The rate function is not restricted to any particular one in order to stress the generality of the proposed approach. 
We compare MPCA with the global optimal solution for LFDMA, as same as in \cite{Lei:2013en}.
Our key contributions are as follows. 

\begin{itemize}
\item We show that the Min-Power for IFDMA is polynomial-time solvable.
\item A polynomial-time algorithm MPCA is developed to solve the Min-Power problem in IFDMA.  
\item Numerically, we find that on Min-Power, LFDMA outperforms IFDMA in maximal supported user demand. When the user demand can be satisfied, LFDMA performs slightly better than IFDMA.
\end{itemize}

This paper is organized as follows. In Section \ref{sec:details}, we introduce system model and Min-Power problem for IFDMA. We further prove that Min-Power is polynomial solvable in IFDMA. The algorithm's description and its pseudo-code are proposed in Section \ref{sec:pseudo-code}. Numerical results are given in Section \ref{sec:comparison}. Section \ref{sec:conclusion} concludes this paper.

\section{Interleaved Min-Power Problem}
\label{sec:details}

\begin{table}[tbp]
  \caption{MATHEMATICAL NOTATIONS}
  \label{tab:math_table}
\centering
  \begin{tabular}{|p{0.1\linewidth}<{\centering}|p{0.8\linewidth}<{\centering}|}
  \hline
   \textbf{notation} & \textbf{description}  \\
    \hline
    $\mathcal{M}$ & the users set \\
    \hline
    $\mathcal{N}$ & the channels set \\
    \hline 
    $M$ & the number of users\\
        \hline
    $N$ & the number of channels\\
        \hline
    $\mathcal{L}_i$ & the $i^{th}$ sub block\\
    \hline
    $\mathcal{J}_i$ & the set of allocated channels to user $i$\\
    \hline
    $\mathcal{B}$ & the set of all channel blocks \\
    \hline
    $\mathpzc{b}_k$ & the channel block identified by $k$ \\
        \hline        
    $c$ & the number of allocated channels for each user\\
        \hline
    $s$ & the interspace size between neighbored sub blocks\\
        \hline
    $q$ & the shift distance of the first allocated channel \\
    \hline
    $L$ & the total length of the channel block \\
    \hline
  \end{tabular}
\end{table}

\subsection{System Model}
\label{subsec:system_model}

Let $\mathcal{M}\triangleq\{1,\ldots,M\}$ and $\mathcal{N}\triangleq\{1,\ldots,N\}$ denote the sets of users and channels, respectively. For uplink, the users in $\mathcal{M}$ send data concurrently to a base station. Each user has a total
power limit, denoted by $P^u$. Moreover, for a user, the power
has to be equal on all allocated channels, subject to a given
channel peak power limit $P^s$. Therefore, a user being allocated $n$ channels will use power at most $\min\left\{\frac{P^u}{n},P^s\right\}$ on each channel. We assume that all users are allocated with the same number of channels. The number of allocated channels to each user is denoted by $c$. The total number of allocated channels is $cM$.
For IFDMA, the channels allocated to users are distributed equidistantly. An example for IFDMA is illustrated in \figurename~\ref{fig:example}. 

The occupied segment of channels is up to three parameters, $c$, $s$ and $q$, shown in \tablename~\ref{tab:math_table}. The parameter $c$ ranges from $1$ to $\lfloor\frac{N}{M}\rfloor$. We let $s$ be the interspace ranging from $0$ to $\lfloor\frac{N-cM}{c-1}\rfloor$. We represent the shift distance from the left end of the channels spectrum as $q$ that ranges from $0$ to $N-L$, where $L=(c-1)\times (M+s)+M$ is the length of the segment (composed of channels and interspaces). We use the term `channel block' to denote the occupied segment of the spectrum, as illustrated in \figurename~\ref{fig:example}. If the parameters $c$, $s$ and $q$ are fixed, then the corresponding channel block is determined. We divide each channel block into $c$ `sub blocks', as $\mathcal{L}_1,\mathcal{L}_2,\ldots,\mathcal{L}_c$ respectively, with $|\mathcal{L}_1|=|\mathcal{L}_2|=\cdots=|\mathcal{L}_c|=M$. The set of all allocated channels for any user $j$ in each $\mathcal{L}$, is denoted by $\mathcal{J}_j$, with $|\mathcal{J}_1|=|\mathcal{J}_2|=\cdots=|\mathcal{J}_M|=c$.   
The total number of different channel blocks is denoted by $K$. We use $\mathcal{B}~(|\mathcal{B}|=K)$ to represent the set of all the $K$ channel blocks, where each element is denoted by $\mathpzc{b}_1, \mathpzc{b}_2, \ldots, \mathpzc{b}_K\in\mathcal{B}$, respectively. 



\subsection{Problem Formalization}

All the $K$ possible channel blocks are obtained by the $\proc{Get-Channels-Sets}$ procedure. The two cases, $c=1$ and $c>1$, should be treated differently. This is because we have only one sub-block when $c=1$. In this case, $s$ is meaningless. In line 2--4, we get the channel blocks with $c=1$. In line 5--9, we get the channel blocks with $c>1$. All the possible channel blocks are saved as $\mathpzc{b}_k$, as shown in line 3 and 8. Then there is one increase on the index variable $k$ for the next iteration, as shown in line 4 and 9. In line 5, the total number of possible channel blocks, $K$, is thus obtained.

\begin{figure}[t]
  \centering
  \includegraphics[width=\linewidth]{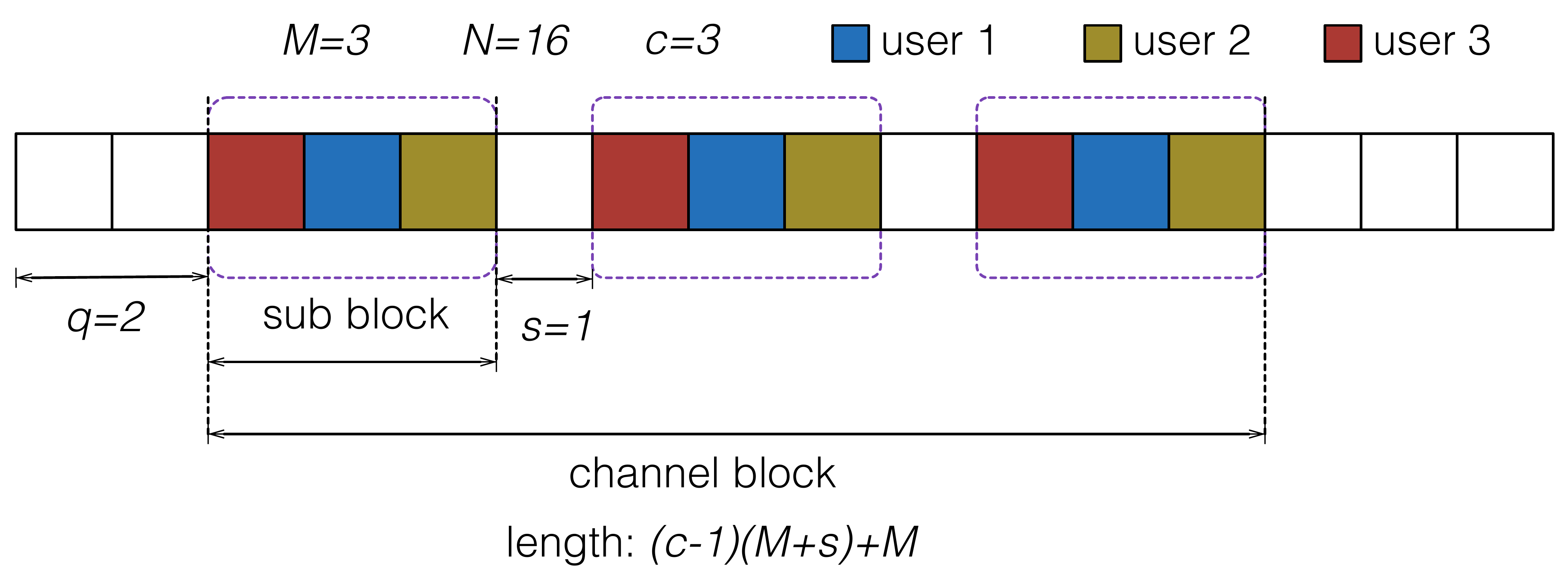}
    \caption{An example for interleaved channel allocation}
  \label{fig:example}
\end{figure}

 We remark that for any channel block $\mathpzc{b}\in\mathcal{B}$, each user is assigned a unique number $j$, representing the $j$th order in each $\mathcal{L}$. The interleaved channel allocation for all the $M$ users is corresponding to a permutation $\bm{\nu}$ of integers $1$ to $M$, i.e., $\bm{\nu}=[\nu_1,\nu_2,\ldots,\nu_M]$, where $\nu_i~(1\leq \nu_i\leq M)$ is the user $i$'s order in each $\mathcal{L}$. In the example shown in \figurename~\ref{fig:example}, the permutation is $\bm{\nu}=[2,3,1]$, which indicates that user $1$ appears in $2$nd order, user $2$ in $3$rd and user $3$ in $1$st order, in all the sub blocks $\mathcal{L}_1,\mathcal{L}_2$ and $\mathcal{L}_3$. 
It can be observed that the allocation for all users is a function of the permutation $\bm{\nu}$, denoted by $\mathpzc{b}(\bm{\nu})$. Together, for each user $i$, $\mathcal{J}_i$ is a function of $\nu_i$, denoted by $\mathcal{J}_i(\bm{\nu})$.

\begin{codebox}
\Procname{$\proc{Get-Channels-Sets}(M,N)$}
\li $c=1,s=k=0$
\li \For $q \gets 0 \To N-cM$
\li \> $\mathpzc{b}_k\leftarrow\textnormal{ascertained by }c,q$
\li \> $k \gets k+1$
\li \For $c \gets 2 \To \lfloor\frac{N}{M}\rfloor$
\li \> \For $s \gets 0 \To \lfloor\frac{N-cM}{c-1}\rfloor$
\li \> \> \For $q \gets 0 \To N-L$
\li \> \> \> $\mathpzc{b}_k\leftarrow\textnormal{ascertained by }c,s,q$
\li \> \> \> $k \gets k+1$
\li $K=k-1$
\li \Return $\mathpzc{b}_1,\mathpzc{b}_2,\ldots,\mathpzc{b}_K$
\end{codebox}

We give the definition of Min-Power problem in IFDMA, where we consider to minimize the total uplink power required to support users' target demand, denoted by $d_i$ for user $i$. We use $f(i,j,p)$ to denote the rate of user $i$ on channel $j$ with power $p$. For the sake of not losing generality, we do not assume any specific power function. Instead, we use $p_{\mathpzc{b}}$ to denote the minimum total power required to satisfy all users' demand on channel block $\mathpzc{b}\in\mathcal{B}$. For each user $i$, the power required to satisfy its demand $d_i$ is represented as $p_i$. Specifically, the power of user $i$ on channel $j$ is denoted by $p_{i,j}$. For that power has to be equal on all channels of user $i$, $p_{i}=\{p: \min \sum_{j\in\mathcal{J}_i}f(i,j,p)\geq d_i\}$, subject to $cp\leq P^u$ and $p\leq P^s$. Thus $p_{\mathpzc{b}}=\min\limits_{\bm{\nu}}\sum_{i\in \mathcal{M}}p_{\mathpzc{b}(\bm{\nu})}$, is the minimal power for the channel block $\mathpzc{b}$. Given $f$, this minimization is straightforward (e.g., bi-section search assuming $f(i,j,p)$ is monotonic in $p$). If the power limits are not exceeded for all users, the allocation is feasible, otherwise the allocation is infeasible. Then the power minimization problem is given, as follows. 


[\textbf{Min-Power}] For each feasible channel-user block $\mathpzc{b}\in \mathcal{B}$ minimizing $p_{\mathpzc{b}}$ by exploring all permutation $\bm{v}=[v_1,v_2,\ldots,v_M]~~(\forall i~,1\leq v_i\leq M)$, where $v_i$ indicates the order for user $i$ in each $\mathcal{L}$ in $\mathpzc{b}$. The optimal power cost is $p^{*}=\min\limits_{\mathpzc{b}_k\in \mathcal{B}}{p_{\mathpzc{b}}}$.


\begin{figure}[t]
  \centering
  \includegraphics[width=\linewidth]{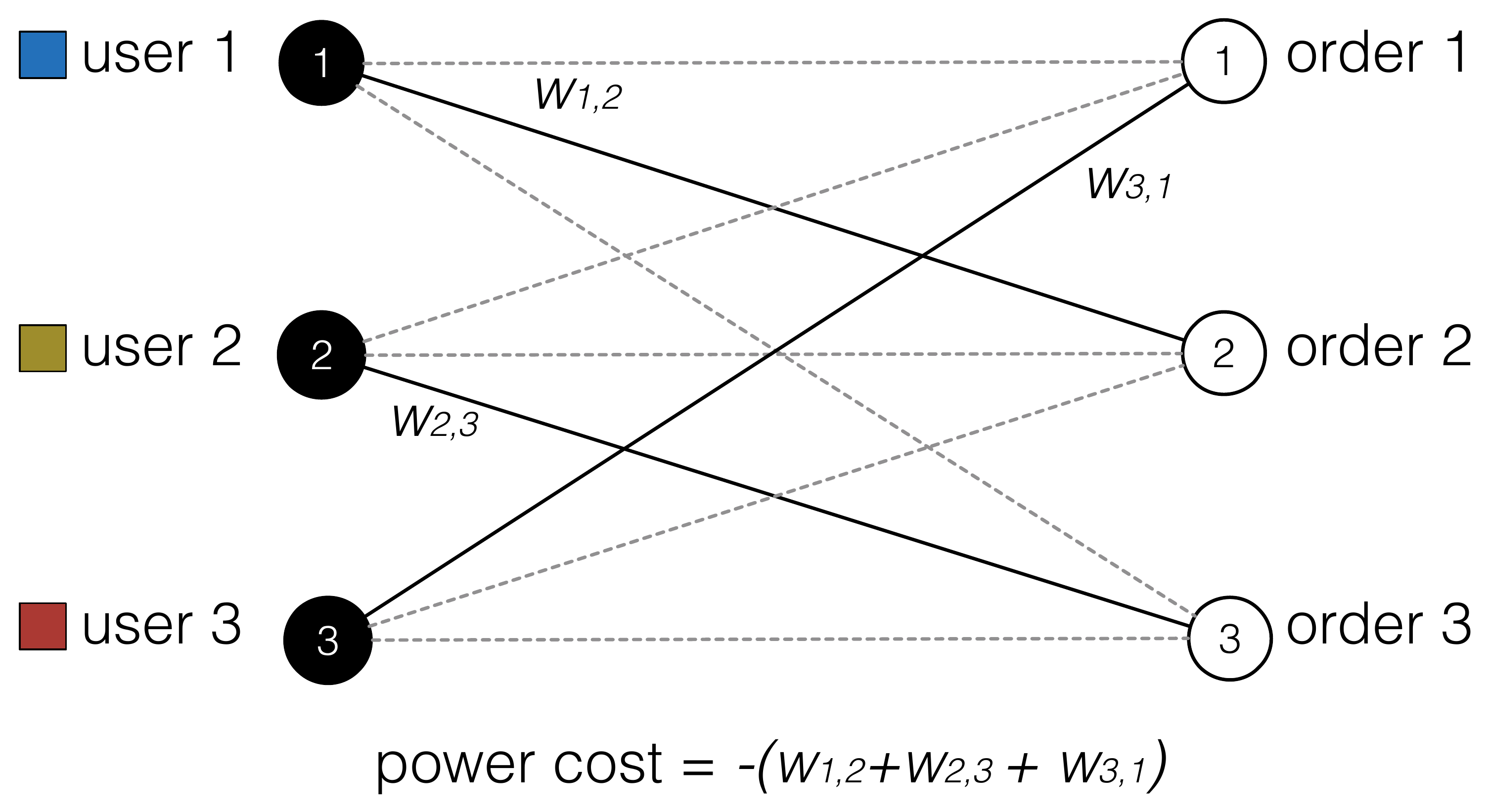}
    \caption{Interleaved Min-Power as a matching problem}
  \label{fig:bipartite}
\end{figure}

Next we show that for each feasible $\mathpzc{b}$, the interleaved Min-Power can be reduced to a maximum weight perfect matching problem.
As illustrated in \figurename~\ref{fig:bipartite}, the nodes in the left part stand for the users, and the nodes in the right part stand for the order in each sub block $\mathcal{L}$. Then each perfect matching in this bipartite graph is corresponding to a permutation $\bm{\nu}=[\nu_1,\nu_2,\ldots,\nu_M]$. More specifically, any left node numbered $i$ is linked with the right node numbered $v_i$ in each matching. The weight for each link between $i$ and $v_i$ is set to be the negative value of the corresponding power cost, respectively. The minimal power cost is equal to the maximal sum of weights among all perfect matchings. The channel allocation problem is thereby identical to the maximum weight perfect matching in a bipartite graph, which can be solved by $\proc{Kuhn-Munkres}$ (KM) algorithm \cite{Lundgren:2010} in $O(M^3)$. The details of KM are given in the appendix.

\subsection{Polynomial-time Proof}
In this section, we prove that for interleaved channel allocation, Min-Power is polynomial-time solvable for global optimality.
\begin{theorem}
For interleaved channel allocation, Min-Power admits
polynomial-time algorithm for global optimality.
\end{theorem}
\begin{proof}
\begin{multline*}
K=(N-cM+1)+(\lfloor\frac{N}{M}\rfloor-1)(\lfloor\frac{N-cM}{c-1}\rfloor+1)(N-L+1) \\
\leq (N+1)+(N+1)^3=O(N^3) ~~~~~~~~(1\leq M \leq N)
\end{multline*}
\vskip 7pt
Then $\proc{Get-Channels-Sets}$ is $O(N^3)$. We can achieve the global optimality by resorting to the enumeration method, i.e., to check every possible channel block $\mathpzc{b}_k\in B$ by running KM once. The total cost for the whole process should be $K\times O(M^3)=O(M^3N^3)$.
\end{proof}

\section{Algorithm description and its pseudo-code}
\label{sec:pseudo-code}

In this section, we give the description and pseudo-code of optimal Min-Power algorithm for the interleaved case, $\proc{MPCA}$.

In $\proc{MPCA}$, firstly the procedure $\proc{Get-Channels-Sets}$ is called to obtain all the $K$ channel blocks $\mathpzc{b}_1, \mathpzc{b}_2, \ldots, \mathpzc{b}_K$. The cost for this part is $O(K)=O(N^3)$. Then the weight for each corresponding matching between any user-channel pair $(i,\nu_i)$ is calculated for every channel block, as shown in line 2---5, of which the cost is $O(KM^2c)$. Since $c\leq \frac{N}{M}$, we have $O(KM^2c)=O(KMN)=O(MN^4)$. Note that we flip the sign for the power cost so as to make sure the Min-Power to coincide with the maximum matching problem in the bipartite graph. The set $Q$ is used to record the past solutions obtained from KM algorithm, and is initialized to be empty at the beginning, in line 6. In line 7---9, the KM procedure is called for all the $K$ possible channel blocks and the total cost is $O(M^3N^3)$. In line 8, the KM returns a two tuple $(p,\bm{\nu})$, where $p$ is the power value and $\bm{\nu}$ is the corresponding permutation. Finally the two-tuple $(p,\bm{\nu})$ that minimizes the value $p$, is returned as the optimal solution. Then the total computation cost of $\proc{MPCA}$ is $O(N^3)+O(MN^4)+O(M^3N^3)=O(M^3N^4)$.

\begin{codebox}
\Procname{$\proc{MPCA}(\mathcal{M},\mathcal{N})$}
\li $\proc{Get-Channels-Sets}(M,N)$
\li \For $k \gets 1 \To K$
\li \> \For $i \gets 1 \To M$
\li \> \> \For $\nu_i \gets 1 \To M$
\li \> \> \> $w_k[i,v_i] \gets -\sum_{j\in\mathcal{J}_i(\bm{\nu})}p_{i,j}$
\li $Q \gets \phi$
\li \For $k \gets 1 \To K$
\li \> $(p,\bm{\nu})=\proc{Kuhn-Munkres}(M,w_k)$
\li \> $Q=Q\bigcup\{(p,\bm{\nu})\}$
\li \Return $\mathop{\arg\min}\limits_{(p,\bm{\nu})\in Q}p$
\end{codebox}
\vskip 10pt

\section{Performance comparison}
\label{sec:comparison}
For performance evaluation, we consider SC-FDMA uplink of a cell with random and uniform user distribution. \tablename~\ref{tab:simulation} summarizes the key parameters. The channel gain consists of path loss, shadowing, as well as Rayleigh fading. The path loss follows the widely used COST 231 model that extends the Okumura-Hata model for urban scenarios. By the COST 231 model, path loss is frequency dependent. Log-normal shadowing model with 8 dB standard deviation is used \cite{Sokmen:2010gv}. A channel corresponds to a resource block in LTE with twelve subcarriers.

\begin{table}
\centering
\caption{SIMULATION PARAMETERS}
\label{tab:simulation}
\begin{tabular}{|l|l|}
\hline
\textbf{Parameter} & \textbf{Value}  \\
\hline
Cell radius & 1000 m \\
\hline
Carrier frequency & 2 GHz \\
\hline
Number of users $M$ & 10 \\
\hline
Number of channels $N$ & 64 \\
\hline
Channel bandwidth $B$ & 180 KHz \\
\hline
Path loss & COST-231-HATA \\
\hline
Shadowing & Log-normal, 8 dB standard \\
          & deviation \\
\hline
Multipath fading & Rayleigh fading \\
\hline
Noise power spectral density & -174 dBm/Hz \\
\hline
User power limit $P^u$ & 200 mW \\
\hline
Channel peak power limit $P^s$ & 10 mW \\
\hline
User demand value $d_i$ & 400--3000 Kbps \\
\hline
\end{tabular}
\vskip 15pt
\end{table} 

We examine two performance aspects. First, performance evaluation of optimal LFDMA \cite{Lei:2013en} and MPCA  (corresponding to IFDMA) for Min-Power has been carried out. Next, we examine the maximal supported user demand of the two allocation schemes. 
We remark that the system model (Section \ref{sec:details}) and the MPCA algorithm are not restricted to any particular definition
of the utility function or power function. For performance
comparison, the power cost value in Min-Power are derived from the logarithmic function. The setting is coherent with the literature \cite{Lei:2013en}, \cite{Sokmen:2010gv} and \cite{Wong:2009gra}. For user $i$, the achieved rate on the whole channels set for the user is $\sum_{j\in\mathcal{J}_i}B\log_2\left(1+\frac{p_ig_{i,j}}{\sigma^2}\right)$, where $B$ is the channel bandwidth, $g_{i,j}$ is the channel gain for user $i$ on channel $j$, and $\sigma^2$ is the noise power spectral density times the channel bandwidth. 
If assigning channels to user $i$ is feasible, the maximal cost of the assignment is clearly $c\cdot\min\left\{\frac{P^u}{c},P^s\right\}$.
So channel assignment is feasible only if the achieved rate over the channels set meets the demand, i.e., if $\sum_{j\in\mathcal{J}_i}B\log_2\left(1+\frac{\min\{P^u/c,P^s\}g_{i,j}}{\sigma^2}\right)\geq d_i$, otherwise the assignment is infeasible and cannot be performed. Uniform demand is used in the simulation, with $d_i=400$ Kbps, $\forall i\in \mathcal{M}$.

\begin{figure}[!tbp]
  \centering
  \includegraphics[width=\linewidth]{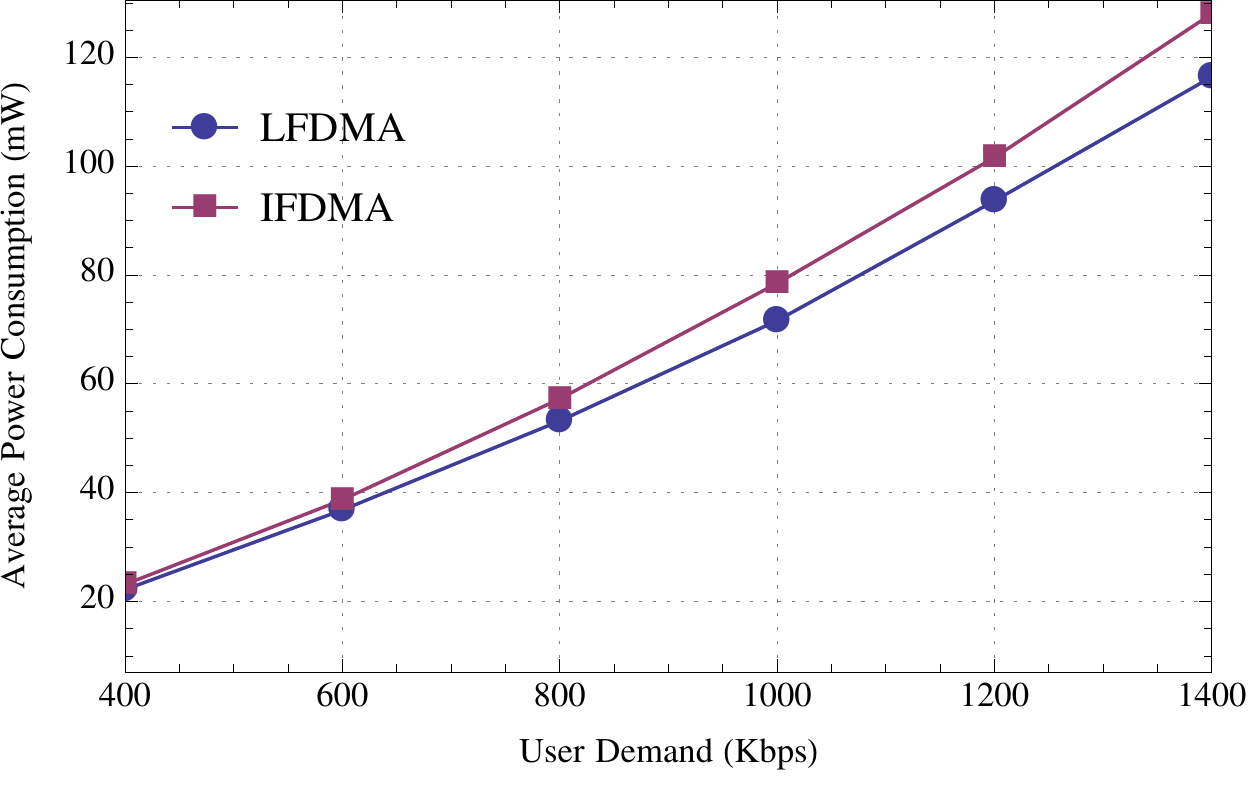}
    \caption{Power consumption with respect to user demand}
  \label{fig:plot1}
\end{figure}

\figurename~\ref{fig:plot1} shows the evaluation results of both LFDMA and IFDMA on varying the user demand. The user power limit $P^u$ is set to be 200 mW and the channel peak power limit $P^s$ is 10 mW. We can see from the results that LFDMA performs slightly better than IFDMA on minimal power consumption. Besides, when the user demand is set to be more than 1.4 Mbps, there is no feasible solution by using IFDMA. However, LFDMA can still support user demand higher than 1.4Mbps. We give the maximal supported user demands results in \figurename~\ref{fig:plot2}.

\begin{figure}[!tbp]
  \centering
  \includegraphics[width=\linewidth]{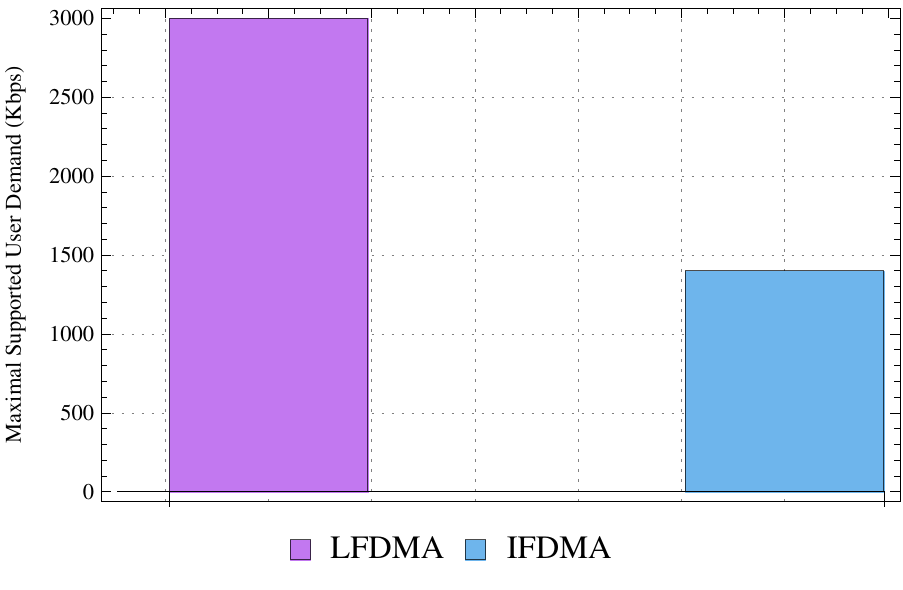}
    \caption{Maximal supported user demand}
  \label{fig:plot2}
\end{figure}

It can be seen from \figurename~\ref{fig:plot2} that the maximal supported user demands are distinctive between LFDMA and IFDMA. The user power limit $P^u$ is constant at 200 mW and the channel peak power limit $P^s$ is set to be 10 mW, as same as in \figurename~\ref{fig:plot1}. Regarding for \figurename~\ref{fig:plot2}, the consecutive channel allocation scheme shows an outstanding performance compared with the interleaved scheme, since the maximal supported power cost for LFDMA is about twice of that for IFDMA.

For Min-Power problem, though the LFDMA outperforms IFDMA in numerical experiments in both \figurename~\ref{fig:plot1} and \figurename~\ref{fig:plot2}, IFDMA is polynomial-time solvable for global optimality. 



\section{Conclusion}
\label{sec:conclusion}

In this paper, we investigated the Min-Power problem on both IFDMA and LFDMA. We proved that Min-Power in IFDMA is polynomial-time solvable. Then the interleaved Min-Power problem was mapped to a maximum weight perfect matching problem in a bipartite graph, which can be solved by resorting to the classic KM approach. The cost of the proposed algorithm MPCA is $O(M^3N^4)$. We numerically compared MPCA with optimal solution in LFDMA. The results showed that for Min-Power, LFDMA outperforms IFDMA in the maximum supported user demand. When the user demand can be satisfied by both LFDMA and IFDMA, LFDMA has slightly better performance than IFDMA. However, Min-Power is polynomial-time solvable for IFDMA whereas it is not for LFDMA.

\section{Acknowledgements}

This work has been supported by the EC Marie Curie
project MESH-WISE (FP7-PEOPLE-2012-IAPP: 324515) and the Link{\"o}ping-Lund Excellence Center in Information Technology (ELLIIT), Sweden.
The work of the second author has been supported by the China Scholarship Council (CSC).

\appendix[Kuhn-Munkres Algorithm for maximum weight perfect matching problem]

\begin{codebox}
\Procname{$\proc{Kuhn-Munkres}(n,w)$}
\li Let $L_x,L_y,V_x,V_y,I,S$ be arrays with length $n$
\li \For $i \gets 1 \To n$
\li \Do $L_x[i] \gets -\infty$
\li     \For $j \gets 1 \To n$
\li     \Do \If $w[i,j]>L_x[i]$
\li         \Then $L_x[i]=w[i,j]$ 
            \End 
        \End
    \End
\li \For $x \gets 1 \To n$
\li \Do \For $i \gets 1 \To n$
\li     \Do $S[i]=\infty$ 
        \End
\li     $found \gets $\textbf{false}
\li     \Repeat
\li         $V_x[0,\ldots,n]=[false,\ldots,false]$
\li         $V_y[0,\ldots,n]=[false,\ldots,false]$
\li         $found \gets \proc{Find-Augment-Route}(x)$
\li         $d \gets \infty$
\li         \For $i \gets 1 \To n$
\li         \Do \If $V_y[i] \isequal $\textbf{false} \textbf{and} $d>S[i]$
\li             \Then $d=S[i]$ 
                \End 
            \End
\li         \For $i \gets 1 \To n$
\li         \Do \If $V_x[i] \isequal $\textbf{true}
\li             \Then $L_x[i]=L_x[i]-d$ \End \End
\li         \For $i \gets 1 \To n$
\li         \Do \If $V_y[i] \isequal $\textbf{true}
\li             \Then $L_x[i]=L_x[i]+d$ 
\li             \Else 
\li                $S[i]=S[i]-d$ 
                \End 
            \End
        \Until $found \isequal $ \textbf{true}
        \End
    \End
\li \Return $\left(-\sum_{i=1}^{i=n}w[I[i],i],I\right)$
\end{codebox}

\begin{codebox}
\Procname{$\proc{Find-Augment-Route}(n,x)$}
\li $V_x[x] \gets $ \textbf{true}
\li \For $y \gets 1 \To n$
\li \Do \If $V_y[y] \isequal $ \textbf{false}
\li     \Then $t \gets L_x[x]+L_y[y]-w[x,y]$
        \End
\li     \If $|t| < \epsilon$
\li     \Then $V_y[y] \gets $\textbf{true}
\li           $found \gets \proc{Find-Augment-Route}(x)$
\li           \If $I[y] \isequal -1$ \textbf{or} $found$
\li           \Then $I[y] \gets x$
\li                 \Return \textbf{true}
              \End
\li      \ElseIf $S[y]>t$
\li      \Then $S[y] \gets t$
         \End
    \End
\li \Return \textbf{false}            
\end{codebox}

%

\end{document}